\documentclass[a4paper,USenglish]{lipics-v2021}

\nolinenumbers\hideLIPIcs

\usepackage{microtype}
\usepackage[T1]{fontenc}
\usepackage[utf8]{inputenc}
\usepackage{hyperref}
\hypersetup{
    colorlinks=true,
    linkcolor=blue
    }
\usepackage{mathtools}
\usepackage{amsfonts}
\usepackage{amsthm}

\newtheorem{fact}[theorem]{Fact}
\newtheorem{obs}[theorem]{Observation}
\usepackage[noend]{algpseudocode}
\usepackage{algorithm}

\hypersetup{
    colorlinks=true,
    linkcolor=blue
    }

\DeclarePairedDelimiter{\abs}{\lvert}{\rvert}
\DeclareMathOperator{\FVD}{\textnormal{FVD}}
\DeclareMathOperator{\polylog}{polylog}
\DeclareMathOperator{\radius}{\opfont{radius}}
\DeclareMathOperator{\width}{width}
\DeclareMathOperator{\cntr}{center}
\newcommand{\TFVD}{T_{\FVD}}
\newcommand{\dF}{d_{\textnormal{dF}}}

\let\eps\varepsilon
\newcommand{\dm}{g} %

\usepackage{graphicx}

\newcommand{\reals}{\mathbb{R}}
\newcommand{\sphere}{\mathbb{S}}

\newcommand{\bydef}{\coloneqq}

\usepackage[dvipsnames]{xcolor}

\newcommand{\opfont}[1]{\texttt{\textup{#1}}}
\newcommand{\disk}{\opfont{disk}}
\newcommand{\ball}{\opfont{ball}}
\newcommand{\prefix}{\opfont{prefix}}
\newcommand{\feasible}{\opfont{feasible}}
\newcommand{\fits}{\opfont{fits?}}

\title{A Dimension-Reducing Fr\'echet Simplification Oracle\footnote{%
  A preliminary version of this paper appeared in the \emph{Proc. 36th International Symposium on Algorithms and Computation} (ISAAC),
                  2025, 6:1--6:20.}}
\titlerunning{A dimension-reducing Fr\'echet simplification oracle}

\author{Boris Aronov}{Department of Computer Science and Engineering, Tandon School of Engineering, New York University, Brooklyn, NY 11201 USA.}
{boris.aronov@nyu.edu}
{https://orcid.org/0000-0003-3110-4702}
{Work partially supported by NSF Grant CCF-20-08551.}

\author{Tsuri Farhana}{The Stein Faculty of Computer and Information Science, Ben-Gurion University of the Negev, Beer Sheva, Israel.}
{tsurif@post.bgu.ac.il}
{https://orcid.org/0009-0006-7246-6805}
{}

\author{Matthew J. Katz}
{The Stein Faculty of Computer and Information Science, Ben-Gurion University of the Negev, Beer Sheva, Israel.}
{matya@bgu.ac.il}
{https://orcid.org/0000-0002-0672-729X}
{Work partially supported by Grant 2019715/CCF-20-08551 from the US-Israel Binational Science Foundation/US National Science Foundation.}

\author{Indu Ramesh}{Department of Computer Science and Engineering, Tandon School of Engineering, New York University, Brooklyn, NY 11201 USA.}{ir914@nyu.edu}{https://orcid.org/0009-0008-9967-0819}{Work supported by NSF Grant CCF-20-08551.}

\authorrunning{B. Aronov, T. Farhana, M.J. Katz, and I. Ramesh}
\Copyright{Boris Aronov, Tsuri Farhana, Matthew J. Katz, and Indu Ramesh}



\begin{document}

\maketitle

\begin{abstract}
    Let $P$ be a polygonal curve with $n$ vertices in the plane. We construct a data structure of size~$O(n \log n)$ suited for simplification queries of the following kind. Given a query line $\ell$ and an integer $k\ge1$, find a curve $Q$ on $\ell$ with at most $k$ vertices that minimizes the discrete Fr\'echet distance to $P$, among all such curves. Using our data structure, a query can be handled in $O(k^2 \log^3 n + k\log^4 n)$ time.  

    More generally, a geometric tree $T$ on $n$ vertices in the plane can be preprocessed into a near-linear-size structure so that, given a pair $u$, $v$ of its vertices, a line $\ell$, and an integer $k\ge1$, one can find a curve $Q$ on $\ell$ with at most $k$ vertices that minimizes the discrete Fr\'echet distance to the path from $u$ to $v$ in~$T$, in time $O(k^2 \polylog n)$. 

    For the general \emph{dimension-reduction} problem,  where $P$ is a curve in $\reals^d$ ($d \ge 3$), $0 < \eps_0 < 1$ is a real parameter, and a query specifies a $g$-flat $h$ ($1 \le g \le d-1$) and an integer $k \ge 1$, we construct a data structure of size $O(n\log n + f(\eps_0) n)$, where $f(\eps_0)=(1+1/\eps_0)^{(d-1)/2}$, that allows us to find a curve $Q$ on $h$ with at most $k$ vertices, whose discrete Fr\'echet distance to~$P$ is at most $1+\eps_0$ times the distance of $Q^*$ to $P$, where $Q^*$ is such a curve that minimizes the distance to $P$. The query handling time is $O(f(\eps_0) k^2 \log^2 n)$.
\end{abstract}

\section{Introduction}

Curve simplification is an important and extensively studied problem in computational geometry and related fields. Here, we restrict our attention to curve simplification with the quality of the simplification measured using the discrete Fr{\'e}chet distance between the original curve and its simplification. See below for a comprehensive survey of research in this area. 

In this paper, we introduce a novel class of curve simplification problems, in which the simplified curve is required to lie within a specified subspace or region. More formally, let $P$ be a polygonal curve in~$\reals^d$. We are interested in problems where, in addition to standard parameters such as a threshold on the distance between~$P$ and the desired simplification~$Q$ or a bound on the length (i.e., number of vertices) of~$Q$, $Q$ must also reside in a designated subspace or region. Moreover, since it is often desirable to find the optimal simplification of~$P$ for multiple subspaces or regions, we concentrate on the multi-shot formulation of these problems.

These questions can also be viewed as \emph{dimension-reduction} problems. That is, we wish to replace an input $d$-dimensional curve by a curve in some lower-dimensional space, while minimizing the Fr\'echet distance to the input curve.

Specifically, we focus on the following problem. 
Let $P$ be a polygonal curve in $\reals^d$. Preprocess $P$ into a compact data structure to efficiently support queries that specify a $g$-flat\footnote{That is, a $g$-dimensional flat---a $g$-dimensional affine subspace of $\reals^d$.} $h$ ($1 \le g \le d-1$) and an integer $k\ge 1$, and seek a curve $Q$ on $h$ of length at most $k$ that minimizes the discrete Fr{\'e}chet distance to~$P$.  For $d=2$, the $g$-flat $h$ is simply a line~$\ell$ in the plane, and we present an efficient solution to the problem, as well as to the more general problem in which the input is a geometric tree $T$ rather than a curve~$P$. In the latter problem, a query also specifies the path in the tree to which the query applies. 
For $d > 2$, we present an efficient solution that provides approximate answers to queries, that is, given $h$ and $k$, it returns a curve $Q$ on $h$ with at most $k$-vertices which is approximately the best such curve on $h$, in terms of its discrete Fr\'echet distance to $P$.   

\subparagraph*{Basic definitions.}
In this paper we view a finite sequence of points $P=(p_1, \ldots, p_n)$ in $\reals^d$ as a polygonal curve; we refer to $n$ as the 
\emph{size} or \emph{length} of $P$ and the points as its \emph{vertices}.

Let $P=(p_1, \ldots, p_n)$ and $Q=(q_1,\ldots,q_m)$ be polygonal curves in $\reals^d$. A \emph{legal walk} of $P$ and $Q$ (or just \emph{walk}, for short) is a monotonically increasing sequence of pairs $(r_1,\ldots,r_l)$, where $r_1=(p_1,q_1)$, $r_l=(p_n,q_m)$, and, in general, when advancing from $r_k=(p_i,q_j)$, $r_{k+1}$ is one of the following: $(p_{i+1},q_j)$ (provided $i<n$), $(p_i,q_{j+1})$ (provided $j < m$), or $(p_{i+1},q_{j+1})$ (provided $i<n$ and $j<m$). The \emph{cost} of a walk is the maximum distance among the distances between the vertices of each pair $r_k = (p_i,q_j)$ in the walk. The \emph{discrete Fr\'echet distance} $\dF(P,Q)$ between $P$ and $Q$ is the cost of the minimum-cost walk of $P$ and $Q$.

Given a $g$-flat $h$, $1 \le g \le d-1$, we say that $Q$ is a \emph{$k$-vertex curve on $h$} if $Q=(q_1,\dots,q_k)$ is a sequence of~$k$ points on $h$. An \emph{at-most-$k$-vertex curve on $h$} is a $k'$-vertex curve on $h$ with $1\leq k' \leq k$.

\subparagraph*{Main problem statement.}

\begin{figure}[ht]
	\centering
	\includegraphics[scale=0.8]{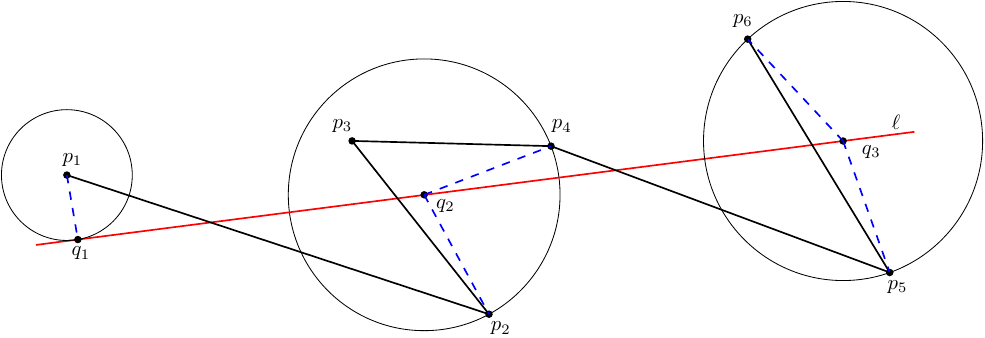}
	\caption{A polygonal curve $P=(p_1,\ldots,p_6)$ and a query line $\ell$.  In this example, $k=3$ and $Q=(q_1,q_2,q_3)$ is the returned curve.}
	\label{fig:line_query}
\end{figure}

In this paper we study the following problem.  Given a curve~$P$ in~$\reals^d$ as above, preprocess it into a near-linear size data structure to efficiently support queries of the following type: 
``Given a $g$-flat $h$, $1 \le g \le d-1$, and integer $k\ge1$, find an at-most-$k$-vertex curve $Q$ on $h$ minimizing $\dF(P,Q)$.''
We first focus on the planar version, where $P$ is a curve in the plane and queries are lines. See Figure~\ref{fig:line_query} for an example.

We also consider the more general setting, still in the plane, where the input is a geometric tree $T$ rather than a curve $P$, and the queries are of the form: ``Given two nodes~$p,q$ of the tree, a line~$\ell$, and integer~$k$, find an at-most-$k$-vertex curve $Q$ on $\ell$ minimizing the discrete Fr\'echet distance between~$Q$ and the path from $p$ to $q$ in $T$.''

In other words, we first study problems in which we need to preprocess a given curve $P$ (or tree $T$) in the plane, so that one can efficiently handle \emph{simplification} queries with respect to a query line. The maximum length $k$ of the desired simplification is also a parameter of the query. 

Next, we further generalize the problem so that a query may specify an arbitrary polygonal region in which $Q$ must reside; query time will of course depend on the complexity of the region.

Finally, we consider the problem in three and higher dimensions, where $P$ is a curve in~$\reals^d$, $d \ge 3$, and queries are $g$-flats, $1 \le g \le d-1$. Since we are interested in a near-linear size data structure, we resort to approximate queries (see discussion in Section~\ref{sec:approx}), with a prespecified error parameter $\eps_0>0$. That is, the answer to a query is an at-most-$k$-vertex curve $Q$ on $h$, such that $\dF(P,Q) \le (1+\eps_0)\dF(P,Q^*)$, where $Q^*$ is such a curve minimizing the distance to $P$.

\subparagraph*{Our results.}
We present efficient solutions to these problems. Specifically, given $P$ or $T$ in the plane, we construct a data structure of size $O(n\log n)$ that allows us to answer a query in $O(k^2 \polylog n)$ time; refer to Theorem~\ref{thm:path} in Section~\ref{sec:path} and Theorem~\ref{thm:tree} in Section~\ref{sec:trees} for the exact statements.
As an intermediate result, we present an algorithm for the corresponding decision problem (with respect to a curve $P$ or tree $T$). For example, given a line~$\ell$ and distance~$r$, we can determine in $O(k \log^2 n)$ time whether there exists an at-most-$k$-vertex curve on $\ell$ within discrete Fr\'echet distance at most $r$ from~$P$, and if so, return such a curve~$Q$ with fewest vertices.
In Section~\ref{sec:region} we extend the data structure so that a query may specify a polygonal region in the plane that restricts the location of the curve $Q$.

In higher dimensions, given $P$ in $\reals^d$ and $0 < \eps_0 < 1$, we construct a data structure of size $O(n\log n + \frac{n}{\eps^{(d-1)/2}})$, where $\eps = \frac{\eps_0}{1+\eps_0}$, that allows us to answer a query approximately in $O(\frac{k^2\log^2 n}{\eps^{(d-1)/2}})$ time, that is, the discrete Fr\'echet distance between $P$ and the curve that we return is at most $1+\eps_0$ times the distance between $P$ and the optimal curve; refer to Theorem~\ref{thm:path_higher} in Section~\ref{sec:approx} for the exact statement.  Our solution uses coresets.

\subparagraph*{Related work on curve simplification.}
We focus here on (global) simplification under the Fr{\'e}chet distance in the plane, where the quality of a simplification is the Fr{\'e}chet distance between the original curve $P$ and its simplification $Q$.

There are two main versions of the simplification problem. In the first, we are given a threshold $\delta$ and the goal is to compute a curve $Q$ of minimum length, such that the distance between $P$ and $Q$ is at most $\delta$. In the second, we are given an integer $k \ge 1$ and the goal is to compute a curve $Q$ of length at most $k$ that minimizes the distance to $P$.
Moreover, each of the two main versions gives rise to four standard variants, depending on whether we are using the continuous or discrete Fr{\'e}chet distance, and whether $Q$ is \emph{restricted} (i.e., a subsequence of $P$) or \emph{unrestricted} (i.e., any sequence of points);
Van~de~Kerkhof et al.~\cite{KerkhofKLMW19} also consider an intermediate \emph{curve-restricted} option.

The following results are for the first version under the continuous Fr{\' e}chet distance. Bringmann and Chaudhury~\cite{BringmannC20} presented an algorithm with running time $O(n^3)$ for restricted simplification and established a conditional cubic
lower bound.
An algorithm with the same running time was also proposed by Van~de~Kerkhof et al.~\cite{KerkhofKLMW19} (see also Van Kreveld et~al.~\cite{KreveldLW20}). 
For unrestricted simplification, Guibas et al.~\cite{GuibasHMS93} provide an algorithm with running time $O(n^2\log^2 n)$, and Agarwal et~al.~\cite{AgarwalHMW05} provide an $O(n \log n)$ algorithm, which returns a simplification $Q'$ of length at most the length
of an optimal simplification with threshold~$\delta/8$. The latter two statements can be deduced from the corresponding papers, as observed by Van~de~Kerkhof et al.~\cite{KerkhofKLMW19}.

Bereg et al.~\cite{BeregJWYZ08} studied both restricted and unrestricted simplification for each of the two main versions, under the discrete Fr{\' e}chet distance in three-dimensional space. For the first version, they compute an optimal unrestricted simplification in $O(n \log n)$ time and an optimal restricted simplification in $O(n^2)$ time. For the second version, they compute an optimal unrestricted simplification in $O(kn \log n \log (n/k))$ time, where $k$ is the length of the computed simplification, and an optimal restricted simplification in~$O(n^3)$ time.
Also under the discrete Fr{\' e}chet distance, Fan et al.~\cite{FanFKWZ15} considered the so called \emph{chain-pair simplification problem} in the plane, where the goal is to simultaneously compute restricted simplifications of length $k$ for two input curves, such that the distance between the simplifications themselves is also below some threshold.

Finally, under the continuous Fr{\' e}chet distance, Driemel and Har{-}Peled~\cite{DriemelH13} presented 
a near-linear time algorithm that computes a permutation of the vertices of $P$, such that any prefix of $2k-1$ vertices of this permutation is an approximation (up to a constant factor) of $P$ compared to any polygonal curve with $k$ vertices. Subsequently, Filtser~\cite{Filtser18} slightly improved the approximation factor under the discrete Fr{\' e}chet distance.

\subparagraph*{More related work.}
A series of papers \cite{constrained-cirle-RoyKDN09,circle-query-ipl-KD08,query-BLR08} culminates in an algorithm that preprocesses an $n$-point set $S$ in the plane in linear space and $O(n\log n)$ time to support $O(\log n)$ time queries of the form ``What is the smallest circle centered on the query line and enclosing~$S$?'' 

There is also a substantial amount of work on so-called \emph{range-aggregate} queries. Given an $n$-point set $S$ in the plane, we want to preprocess it for queries of the following type: compute some aggregate quantity about the points of $S$ contained in an axis-aligned query rectangle.  While range counting, reporting, and more general semigroup queries have been widely studied~\cite{agarwal2017range}, some types of queries are not readily decomposable (in the sense of the answer being easily synthesized from the answer for smaller queries).  For example,  Brass~et~al.~\cite{aggr13} discuss computing the area (or perimeter) of the convex hull, width, and smallest enclosing disk of the points of $S$ in the query rectangle.  They describe algorithms for preprocessing $S$ into a data structure of near-linear size with polylogarithmic-time queries, for all of the above queries except for the width; the width data structure is more expensive. Gupta~et~al.~\cite{aggr14} explore several variants of closest-pair-in-range problems.  For example, they show a near-linear-space structure for computing the closest pair of points within the query rectangle in two dimensions with polylogarithmic query time. Our observation in Remark~\ref{rem:aux_result} is a further generalization in this direction: Given a point set in the plane, we can preprocess it so that, given a query rectangle \emph{and} a query line, we can determine the smallest enclosing circle of the points in the rectangle, among all circles centered on the line.

\subparagraph*{Overview and organization.}
We first consider the planar version of the dimension-reduction problem. To answer a \emph{decision} query of the form ``given a line $\ell$, an integer $k$, and a distance $r > 0$, determine whether there exists an at-most-$k$-vertex curve on $\ell$ within Fr\'echet distance at most $r$ from $P$,'' we need to repeatedly perform the primitive operation $\prefix_\ell(P[i,n],r)$ for different values of $i$: Find the longest (contiguous) subcurve of $P$, beginning at $p_i$, that fits into a disk of radius $r$ centered at a point of $\ell$. In Section~\ref{sec:factsandtools}, we describe our data structure for the planar version and then use it to implement $\prefix_\ell(P[i,n],r)$ in $O(\log^2 n)$ time. Our data structure is obtained by combining known data structures in a non-trivial manner. In Section~\ref{sec:path}, we obtain our first main result, i.e., the solution for the planar version of the dimension-reduction problem. Using the data structure and primitive operations developed in Section~\ref{sec:factsandtools}, we present efficient algorithms for handling both decision queries and `regular' queries, where the algorithm for regular queries is based on the one for decision queries and on the primitive operation $\radius_\ell(P[i,j])$ (developed in Section~\ref{sec:factsandtools}) which, given a (contiguous) subcurve of $P$, returns the radius of the smallest disk centered at a point of $\ell$ that contains the subcurve. In Section~\ref{sec:approx}, we obtain our second main result, i.e., the solution for the $d$-dimensional version, $d \ge 3$, of the dimension-reduction problem. The algorithms in this general version are similar to those in the planar version, however, we need to replace the exact primitive operations above by their approximate $d$-dimensional counterparts.
These latter operations are implemented using coresets.
Finally, in Sections~\ref{sec:trees}~and~\ref{sec:region}, we extend our result for the planar version in two directions. In Section~\ref{sec:trees}, the input is a geometric tree~$T$ (rather than a curve~$P$), and queries specify, in addition to $\ell$ and $k$, a path in $T$, by providing its start and end vertices. In Section~\ref{sec:region}, queries specify a polygonal region rather than a line $\ell$, to contain the simplified curve.

\section{Facts and tools}
\label{sec:factsandtools}
\subsection{Facts}
\label{sec:facts}

For two points $p,q$ in the plane, let $d(p,q)$ be the Euclidean distance between them; many of our observations generalize to other norms, but in this version we focus on the Euclidean metric, for simplicity of presentation.

For a point set $Q$, let $\radius(Q)$ be the radius of the \emph{minimum enclosing circle} (or \emph{mec}, for short) of $Q$, and let $\radius(Q,p)$ be the radius of the mec of $Q$ centered at $p$.
Moreover, let $\radius_\ell(Q)$ be the radius of the mec of $Q$ centered at a point of line $\ell$, and
let $\radius_\ell(Q,x)$ be the radius of the mec of $Q$ centered at $x\in \ell$, as a function of $x$.
We give the proof of the following standard fact below, for completeness.

\begin{fact}
  \label{fact:convex}
  Both $\radius(Q,x)$ and $\radius_\ell(Q,x)$ are convex and have a unique minimum.
\end{fact}

\begin{proof}
  Since $\radius(Q,x)=\max_{q\in Q} d(q,x)$, and $d(q,x)$ is a convex function of $x$, $\radius(Q,x)$ is convex.  Therefore  $\radius_\ell(Q,x)$ is also convex, as a restriction of the convex function $\radius(Q,x)$ to a line.
  
  As for uniqueness of minimum, we prove it for $\radius_\ell(Q,x)$; it also proves the claim for $\radius(Q,x)$.

  For a contradiction, suppose $p$ and $q\neq p$ are both minima of $\radius_\ell(Q,x)$. Then the disks~$D_p,D_q$ of radius $r\bydef \radius_\ell(Q,p)=\radius_\ell(Q,q)$ centered at $p$ and $q$, respectively, each cover $Q$.  Therefore $Q$ is contained in the lens $D_p\cap D_q$.  Let $s$ be a point on the open segment $pq\subset\ell$.  Then the disk centered at $s$ and just covering the lens $D_p\cap D_q$ contains~$Q$ and has a smaller radius, contradicting the assumption that $p$ and $q$ were minima of $\radius_\ell(Q,x)$.
\end{proof}

\subsection{Primitives in the plane}
\label{sec:prim}
We now list the operations that will be useful in solving the two-dimensional problems addressed in this paper. 
Each of the operations below depends on the query line $\ell$. %
\begin{description}
\item{$\disk_\ell(Q)$}: The smallest-radius disk centered at a point of $\ell$ and containing a set $Q$ of points; let $\radius_\ell(Q)$ and $\cntr_\ell(Q)$ denote its radius and center, respectively.
\item{$\feasible_\ell(Q,r)$}: The interval of $\ell$ (which may be empty) that is the locus of centers of disks of radius $r$ containing $Q$.
\item{$\fits_\ell(Q,r)$}: Is there a disk of radius $r$ centered at a point of $\ell$ and containing a set $Q$ of points?  Equivalent to $\feasible_\ell(Q,r)\neq \emptyset$ or to $\radius_\ell(Q)\leq r$; the former version is sometimes more efficient.
\item{$\prefix_\ell(S,r)$}: For a sequence $S$ of points, its longest prefix (which may be empty or all of~$S$) that fits into a disk of radius $r$ centered at a point of $\ell$.  In practice, $\prefix_\ell(S,r)$ returns the length of the prefix, so 0 if it is empty. 
\end{description}

\subsection{Data structures in the plane}
\label{sec:data}

We now describe the data structures we employ to implement the above operations. Consider a polygonal curve $P$ of length~$n>1$. For clarity of presentation we assume that $n$ is a power of two; the general case requires minimal straightforward modifications. We denote by
$P[i,j]$ the (contiguous) subcurve $(p_i,p_{i+1},\dots,p_j)$, for $1 \leq i \leq j \leq n$.

\subparagraph*{The hierarchical decomposition.}
 The \emph{hierarchical decomposition} of $P$ is a binary tree, in which the root corresponds to $P=P[1,n]$, its left and right children to $P[1,n/2]$ and $P[n/2+1,n]$, respectively, and so forth. The leaves correspond to the single-vertex curves.  Each node of the tree corresponds to a \emph{canonical subcurve} of~$P$ (we will use the term ``canonical set'' when the ordering of the vertices in the subcurve is immaterial).
A general subcurve $P[i,j]$ corresponds to the concatenation of $O(\log n)$ canonical subcurves, which can be identified from the value of the indices $i$~and~$j$ in $O(\log n)$ time.

We often build data structures based on the hierarchical decomposition of a curve, where at each node we store an additional structure tailored to its corresponding canonical subcurve.

\subparagraph*{Farthest-point Voronoi diagram as a map.}
For a given set of points $Q$ in the plane, the \emph{farthest-point Voronoi diagram} $\FVD(Q)$ is a planar map decomposing the plane into convex regions $R(p)$, one for each $p\in Q$, where $R(p)$ is the locus of all points in the plane for which $p$ is the farthest point in $Q$.  $\FVD(Q)$ is a planar map with $O(|Q|)$ vertices, edges, and faces.  Each face is a convex unbounded polygon.  We preprocess the planar map for point location, so that, given a query point~$q$, one can find the region $R(p)$ containing~$q$, and thereby the farthest point of~$q$ in~$Q$, in $O(\log|Q|)$ time.

\subparagraph*{Centroid decomposition.}
Consider a free tree $T$ of degree at most three.
A \emph{centroid edge} is an edge of~$T$ whose removal splits~$T$ into two trees of almost equal size, with between $1/3$ and $2/3$ of the vertices of $T$. It is known to exist in any degree-at-most-three tree. The existence of a centroid edge naturally induces a \emph{centroid decomposition} of $T$: the root is a centroid edge~$e$. Each subtree left and right of the root is, recursively, a centroid decomposition of the two trees formed by the removal of $e$ from $T$.  Each internal node corresponds to an edge of $T$. Each leaf corresponds to either a single edge or to two edges of $T$.

\subparagraph*{Farthest-point Voronoi diagram as a tree.}
Given a set $Q$ in the plane, the 1-skeleton of~$\FVD(Q)$ (i.e., the collection of its vertices and edges, viewed as a graph with virtual vertices at the ``ends'' of infinite ray edges) is a tree, which, for $Q$ in general position, has degree three.
For simplicity of exposition, we will assume that the points of $Q$ are in general position; refer to~\cite{query-BLR08} for how to handle degeneracies in $Q$ when constructing the centroid decomposition.

We now describe what is essentially the centroid-decomposition-based data structure from~\cite{query-BLR08} for computing the smallest disk enclosing a given set $Q$ of points, with the center on the query line~$\ell$, i.e., $\disk_\ell(Q)$. For a fixed set~$Q$, the centroid decomposition of $\FVD(Q)$ requires linear space and $O(\abs{Q}\log\abs{Q})$ time to build. In our data structure, for a path $P$, we build the centroid decomposition for all canonical subsets of~$P$, which requires $O(n\log n)$ overall space and $O(n\log^2n)$ time.

We store the centroid decomposition of $\FVD(Q)$ together with some additional geometric information, as follows: Each internal node~$x$ corresponds to an edge $e_x$ of $\FVD(Q)$, which is a line segment separating regions $R(v_x)$ and $R(w_x)$. Let $m(e_x)$ be an arbitrary point of $e_x$.  There exist two infinite rays  $\rho(v_x)$ and $\rho(w_x)$\cite{query-BLR08} emanating from $m(e_x)$, one fully contained in $R(v_x)$ and one in $R(w_x)$, respectively.  Store these two rays together with the centroid edge $e_x$ and points $v_x$ and $v_w$ at the current node~$x$ of the centroid decomposition tree.  
Notice that the polygonal line $\rho(v_x)\cup \rho(w_x)$ splits the plane into two regions: one contains the subtree corresponding to the left child of $x$ and the other --- the subtree of its right child.

\subparagraph*{The full structure $\TFVD$.}
Let $\TFVD=\TFVD(P)$ be the data structure constructed as follows: We start with the hierarchical decomposition of $P$. At each node of the decomposition, which corresponds to a canonical subcurve $P[i,j]$ of $P$, we store $\FVD(P[i,j])$ both in the planar map and centroid decomposition form.

\subsection{Implementations}
\label{sec:impl}

We now describe how to efficiently implement the primitives and analyze their running times.

\subsubsection{Implementation of \texorpdfstring{$\disk_\ell(Q)$}{disk(Q)}}
\label{sec:impl-disk}

\subparagraph*{If $Q$ is a canonical set of size $n$.}
The algorithm in \cite{query-BLR08} preprocesses $Q$ in linear space and $O(n\log n)$ time to support such an operation in  $O(\log n)$ time; see Section~\ref{sec:data} for a description of the centroid decomposition data structure.  Roughly speaking, 
the search proceeds by descending the centroid decomposition of $\FVD(Q)$, narrowing down the portion of the query line $\ell$ containing the disk center, at a constant cost per level.   At a leaf of the decomposition, at most three farthest neighbors in~$Q$ remain, which allows to identify the center and radius directly.  We use a more elaborate version below to solve the two-set version of the problem.

\subparagraph*{If $Q$ is a union of two canonical sets.}

Let $Q_1$ and $Q_2$ be two canonical sets of total size at most $n$, and let $\ell$ be the query line.  Put $Q \bydef Q_1 \cup Q_2$.  We wish to compute $\disk_\ell(Q)$, the smallest disk with center on $\ell$ enclosing $Q$; let $x^*$ denote its center, to be computed.  We will need the following information about the sets $Q_1$ and $Q_2$ (already stored in $\TFVD$): $\FVD(Q_1)$ and $\FVD(Q_2)$ viewed both as planar maps, preprocessed for point location, and as planar trees, stored as centroid decompositions, see Section~\ref{sec:data}. 

For clarity of explanation, identify $\ell$ with the $x$-axis and let $x$ denote a generic point of $\ell$.  The center of the desired disk is the point $x=x^*$ minimizing $\radius_\ell(Q,x) = \max(\radius_\ell(Q_1,x),\radius_\ell(Q_2,x))$. 
Each of the three functions are convex on $\ell$, with a unique minimum, by Fact~\ref{fact:convex}.

Note that the following \emph{sidedness test} can be performed in $O(\log n)$ time: Given a point $x \in \ell$, we can determine whether $x^*$ lies before, after, or at~$x$ along~$\ell$.  Indeed, after querying with $x$ in $\FVD(Q_1)$ and $\FVD(Q_2)$, 
we can determine the distance from $x$ to the farthest points of $Q$ (possibly more than one) and therefore the direction in which this distance decreases along~$\ell$---that's the direction towards $x^*$; if no such direction exists, then $x=x^*$.  The bottleneck of this operation is the  lookup of $x$ in the planar map $\FVD(Q_1)$ and $\FVD(Q_2)$, at a cost of $O(\log n)$.

Finally, we describe the computation of $\disk_\ell(Q)$.
Throughout the search we maintain the current segment $s \subset \ell$ that is guaranteed to contain $x^*$.  We initially set $s \bydef \ell$.

We now descend the centroid decompositions of $\FVD(Q_1)$ and $\FVD(Q_2)$.  Without loss of generality, suppose we are currently descending the decomposition of $\FVD(Q_1)$, with the current node corresponding to edge $e$ and its two rays $\rho_1,\rho_2$ splitting the plane into wedges $W_1$ and $W_2$; $s$ intersects $\rho_1\cup \rho_2$ at most twice.
At each of the intersection points we perform our sidedness test and eliminate part of $s$ as the possible position of $x^*$.  The result is an, in general, smaller updated segment $s$, contained fully in~$W_1$ or in~$W_2$ and containing $x^*$.  We descend the decomposition of $\FVD(Q_1)$ to the corresponding child of $e$.

Before terminating the current step, we check if the updated $s$ crosses $e$.  If it does, we split $s$ at the intersection point, check what side contains $x^*$ by another sidedness test, and shrink $s$ further.  Notice that the resulting segment $s$ is guaranteed not to meet any of the edges of the discarded child of $e$ in the centroid decomposition, nor $e$ itself.
Thus we inductively maintain the invariant that the current segment $s$ can only intersect the edges of $\FVD(Q_i)$ corresponding to the current subtree in its centroid decomposition.

We continue in this manner, until we reach the bottom of both hierarchies. 
At a single-edge leaf (a two-edge leaf is handled similarly) of a centroid decomposition of, say, $\FVD(Q_1)$, one edge remains, and we apply the same trick to shrink the segment $s$ to only one side of the edge, belonging to, say the Voronoi region of $q_1 \in Q_1$. 
Similarly, we narrow things down to one site $q_2 \in Q_2$. By construction, now $x^* \in s$, and $s$ intersects no edges of $\FVD(Q_1)$ nor of $\FVD(Q_2)$.  In particular, for all points on $s$, $q_1$ is the furthest point of $Q_1$ and $q_2$---of $Q_2$.  We check one of three possibilities: (a)~$q_1$ could be the only point determining $\disk_\ell(Q)$: we project $q_1$ orthogonally to $\ell$ and check if its projection~$q'_1$ lies in $s$.  If so, check that $Q$ is contained in the disk of radius $d(q_1,q'_1)$ centered at~$q'_1$ (i.e., that $d(q_2,q_1') \le d(q_1,q_1')$)---then this is $\disk_\ell(Q)$. (b) We check if $q_2$ is the only point defining $\disk_\ell(Q)$ similarly.  Finally, (c) we construct the intersection of the bisector $b(q_1,q_2)$ with $s$ (or, equivalently, with $\ell$).  This point~$q$ is the center of $\disk_\ell(Q)$ and its radius is $d(q,q_1)=d(q,q_2)$.  It is easily checked that the processing at the bottom of the two hierarchies takes $O(1)$ work once the identity of $q_1$ and $q_2$ is known.

Since the descent of the two logarithmic-height hierarchies takes $O(\log n)$ time for one step (the bottleneck in the sidedness test is consulting a point with known position in one of the FVDs for its furthest point in the other set: in each sidedness test we already know where we are in one diagram, but not in the other), we reach the bottom and the solution in $O(\log^2 n)$ time.  We conclude that

\begin{lemma}
  \label{lem:union_of_2}
  The operation $\disk_\ell(Q_1 \cup Q_2)$, where $Q_1, Q_2$ are two canonical sets of total size $n$, can be performed in $O(\log^2 n)$ time. 
\end{lemma}

\subparagraph*{If $Q$ is a union of $m \ge 2$ canonical sets.}

\begin{lemma}
\label{lem:union_of_m}
  Let $Q_1, \ldots, Q_m$ be $m\ge2$ canonical subsets of $P$ of total size at most $n$, let their union be $Q$, and let $\ell$ be a line.   We can compute $\disk_\ell(Q)$ in $O(m^2 \log^2 n)$ time.
\end{lemma}

\begin{proof}
  For each pair $(i,j)$, with $1 \le i < j \le m$, we compute $\disk_\ell(Q_i\cup Q_j)$ by an application of the two-set procedure above and return the largest disk.
  The runtime bound is immediate, so we focus on correctness.

  Since $\disk_\ell(Q)$ is determined by at most two points of~$Q$ (refer to \cite[Observation~1]{constrained-cirle-RoyKDN09}), it follows that the desired disk is indeed $\disk_\ell(Q_k \cup Q_l)$, if one of the defining points comes from $Q_k$ and the other from $Q_l$ with $l\neq k$; if the two points come from the same set $Q_k$, then for every $l\neq k$, $\disk_\ell(Q_k\cup Q_l)=\disk_\ell(Q)$.  If $m=2$, we are done.  Otherwise, notice that only the largest of the disks $\disk_\ell(Q_i\cup Q_j)$, $i < j$, can be $\disk_\ell(Q)$, as any smaller disk would not enclose $Q_k \cup Q_l$, by definition of $\disk_\ell(Q_k\cup Q_l)$ and Fact~\ref{fact:convex}.
\end{proof}

\begin{remark}
\label{rem:aux_result}
    Lemma~\ref{lem:union_of_m} is interesting in its own right, since it allows us to solve problems such as the following. Let $S$ be a set of $n$ points in the plane. Construct a data structure of near-linear size to support, in $O(\polylog n)$ time, queries of the form: Given an axis-aligned rectangle $R$ and a line $\ell$, return the smallest enclosing circle of $S \cap R$ with center on $\ell$.    
\end{remark}

\subsubsection{Implementation of \texorpdfstring{$\feasible_\ell(Q,r)$}{feasible(Q,r)}}
\label{sec:feasible}

\subparagraph*{If $Q$ is a single canonical set.}
  Let $s\bydef \feasible_\ell(Q,r)$ be the desired locus of candidate centers $x\in \ell$ such that $\radius_\ell(Q,x) \leq r$; it is a contiguous interval on $\ell$, if non-empty, since $\radius_\ell(Q,x)$ is convex; see Fact~\ref{fact:convex}. This interval is empty if $r < \radius_\ell(Q)$, consists of a single point if $r = \radius_\ell(Q)$, and is a positive-length interval with two endpoints if $r > \radius_\ell(Q)$. 
  Since $\radius_\ell(Q,x)$ is convex and monotone along each of the two half-lines into which $l$ is split by~$\cntr_\ell(Q)$, 
  we can search along each half-line to determine the endpoints of $s$. Using the same data structure as in \cite{query-BLR08} allows us to conduct the search in $O(\log n)$ time. In particular, we descend the centroid decomposition to find a point on each half-line where the radius function switches from being smaller than~$r$ to larger than~$r$. 

\subparagraph*{If $Q$ is the union of $m$ canonical sets $Q_1,\dots,Q_m$.}
Compute $s_i\bydef\feasible_\ell(Q_i,r)$, for $i=1,\dots,m$, and return $\bigcap s_i$.  
Running time is dominated by the cost of computing the $m$ feasible intervals, so it is $O(m\log n)$.

\subparagraph*{If $Q$ is the union of a dynamic collection of canonical sets.}
The goal is to maintain the feasibility interval subject to addition and removal of canonical sets. 
For each set, store its feasible interval and, using two heaps $H_L$ and $H_R$, maintain the rightmost left endpoint and leftmost right endpoint of the intervals.
(i) On insertion of a new set $Q'$, compute $[a_i,b_i]\bydef\feasible_\ell(Q',r)$, add $a_i$ to $H_L$, and $b_i$ to~$H_R$.
(ii) On deletion of $Q'$, remove the endpoints of the stored interval $\feasible_\ell(Q',r)$ from $H_L$ and $H_R$.
(iii) The current feasible interval is always $[\max(H_L),\min(H_R)]$, unless $\max(H_L)>\min(H_R)$, in which case it is empty.

If we use a standard binary heap to implement $H_L$ and $H_R$, insertion costs $O(\log n + \log m)=O(\log n)$, deletion $O(\log m)$, and current feasible region is available in time $O(1)$, where $m$ is the current number of sets in~$Q$ and $n$ is the total number of points~involved.

\subsubsection{Implementation of \texorpdfstring{$\fits_\ell(Q,r)$}{fits?(Q,r)}}
This test is logically equivalent to $\feasible_\ell(Q,r)\neq \emptyset$ or to $\radius_\ell(Q)\leq r$.  The former is more efficient unless $Q$ is a single canonical set (in which case the latter is more straightforward, but equally efficient asymptotically).

\subsubsection{Implementation of \texorpdfstring{$\prefix_\ell(S,r)$}{prefix{S,r}}}
\label{sec:impl-prefix}
This operation can be implemented as a binary search along~$S$, using $\fits(S',r)$ as a blackbox for candidate prefixes~$S'$ as the decision procedure, at the cost of a multiplicative logarithmic factor in running time.
However, sometimes we can do better, as detailed below.

\subparagraph*{If $S$ is a canonical subcurve.}
At any moment during a binary search in the hierarchical decomposition of~$S$ (which itself is a canonical subcurve of~$P$), there are $O(\log n)$ canonical sets $S_1,\dots, S_i$ comprising the current prefix being tested. As binary search progresses down the hierarchy, either a new canonical set $S_{i+1}$ is added to the current prefix (if the prefix is too short), or $S_i$ is replaced by a new, smaller $S_i'$ (if the prefix is too long).
We use the dynamic structure for feasible interval maintenance from Section~\ref{sec:feasible} to keep track of this information and check if the current feasible interval is empty. Since the total number of insertions into and deletions
from our collection of canonical sets is $O(\log n)$,
the overall cost is $O(\log^2 n)$.

\subparagraph*{If $S$ is a subcurve of $P$.} 
 Using the hierarchical decomposition of $P$, we express $S$ as a concatenation $S_1\cdot\ldots\cdot S_m$ of $m=O(\log n)$ canonical subcurves $S_1,\dots,S_{m}$.  In $O(m \log n)$ time, we compute the feasible intervals $s_1,...,s_{m}\subset \ell$ for all subcurves $S_i$. 
 Then in time $O(m)$, we compute the largest index $j$ such that $S_1\cdot \ldots \cdot S_j$ fits into a disk of radius $r$, but $S_1\cdot \ldots \cdot S_{j+1}$ does not, by computing the largest $j\le m$ such that $\bigcap_{1\leq i\leq j} s_i \neq \emptyset$ (if $j=m$, $\prefix_\ell(S,r)$ returns all of $S$). We then binary search within~$S_{j+1}$ to identify the longest prefix of $S_{j+1}$ that can be concatenated to $S_i\cdot \ldots \cdot S_j$ and still fit into a disk of radius $r$ in $O(\log^2 n)$ time, using a slight modification of the above single-canonical-subcurve algorithm. The overall time complexity is $O(m \log n + m + \log^2 n)=O(\log^2 n)$.

\section{Discrete Fr\'echet distance simplification on a line}
\label{sec:path}

In this section we present an algorithm for answering queries, which refer to the already preprocessed input curve $P=(p_1,\ldots,p_n)$, of the following type:
Given a line $\ell$ and an integer $k \ge 1$ find an
at-most-$k$-vertex curve $Q$ on $\ell$ minimizing $\dF(P, Q)$; that is, find a curve $Q =(q_1,\ldots,q_{k'})$, where $k' \le k$ and $q_1,\ldots,q_{k'} \in \ell$, that realizes the expression
\begin{equation*}
\min\limits_{\substack{k'\in [1,k] \\ q_1,\ldots,q_{k'}\in \ell}}
    \{\dF(P,(q_1,\ldots,q_{k'}))\}.
\end{equation*}

We assume that $k<n$, since otherwise the curve $Q=(\overline{p}_1,\ldots,\overline{p}_n)$ clearly achieves the minimum, where $\overline{p}_i$ is the orthogonal projection of $p_i$ on $\ell$.

Recall that we write $P[i,j]$, for $1 \le i \le j \le n$, to denote the (contiguous) subcurve $(p_i, \ldots, p_j)$ of $P$. For a point $q$ in the plane, the distance from $q$ to the vertex of $P[i,j]$ farthest from (nearest to) it, is denoted by $d_{\max}(P[i,j], q)$ ($d_{\min}(P[i,j], q)$). Notice that
\begin{equation*}
\min\limits_{q\in\ell}\{\dF(P[i,j],(q))\} = 
\min\limits_{q\in\ell}\{d_{\max}(P[i,j], q)\} = 
\radius_\ell(P[i,j]),
\end{equation*} 
where $\radius_\ell(P[i,j])$ is the radius of $\disk_\ell(P[i,j])$ (see Section~\ref{sec:prim}).

Let $Q=(q_1,\ldots,q_{k'})$ be the desired curve and let $W$ be a walk of $P$ and $Q$ of cost $\dF(P,Q)$. We may assume that 
$W$ does not match two consecutive points in $Q$ to the same point in $P$. That is, $W$ does not contain two consecutive pairs of the form $(p_i,q_j),(p_i,q_{j+1})$, since, if it does, we may delete the pair $(p_i,q_{j+1})$ from $W$ (and the point $q_{j+1}$ from $Q$, if it was only matched to $p_i$) to obtain a legal walk $W'$ (and a curve $Q'$) of the same cost.

This assumption implies that, in order to compute $\dF(P,Q)$, we need to find an optimal partition of $P$ into $k'$ subcurves. That is, 
\begin{equation*} 
        \dF(P,(q_1,\ldots,q_{k'}))=
        \min\limits_{1\le i_1< i_2<\dots<i_{k'-1}<n}
        \max\left\{
        \begin{gathered}
        d_{\max}(P[1,i_1],q_1)\\
        \vdots\\
        d_{\max}(P[i_{k'-2}+1,i_{k'-1}],q_{k'-1})\\ d_{\max}(P[i_{k'-1}+1,n],q_{k'})
        \end{gathered}
        \right\}.
\end{equation*}
We conclude that, given $\ell$ and $k$, our task is to find $k' \le k$ and a sequence of indices $1\le i_1< i_2<\dots<i_{k'-1}<n$ that minimize the expression  
\begin{equation*}
\max\{
\radius_\ell(P[1,i_1]),\ldots,\radius_\ell(P[i_{k'-2}+1,i_{k'-1}]),\radius_\ell(P[i_{k'-1}+1,n])\}.
\end{equation*}
The desired sequence is then $Q=(\cntr_\ell(P[1,i_1]),\ldots,\cntr_\ell(P[i_{k'-1}+1,n]))$, where $\cntr_\ell(P[i,j])$ is the center of $\disk_\ell(P[i,j])$.

\subsection{The Algorithm}
Given the supporting data structures and set of primitive operations, the query algorithm is pretty simple. We first present an algorithm for the decision problem, i.e., given $\ell$, $k$, and $r$, determine if there exists an at-most-$k$-vertex curve $Q$ on $\ell$ such that $\dF(P, Q) \le r$.

\subparagraph*{The decision algorithm.} 
The function \textsc{Decision} in Algorithm~\ref{alg:dec-k-points} uses a simple greedy recursive approach for solving the decision problem.

\begin{algorithm}[h]
\begin{algorithmic}
    \Function{Decision}{$i,\ell,k,r$}
        \If{$i=n+1$}
            \State return \textsc{true}
        \EndIf 
        \If{$k=0$}
            \State return \textsc{false}
        \EndIf
        \State $l\gets \prefix_\ell(P[i,n],r)$
        \If{$l=0$}  \Comment{$p_i$ is too far from $\ell$}
            \State return \textsc{false}
        \EndIf
        \State return \Call{Decision}{$i+l,\ell,k-1,r$}
    \EndFunction
    \end{algorithmic}
\caption{The decision algorithm: Given a distance $r > 0$, determine if there exists a sequence of at most $k$ points on line $\ell$ at discrete Fr\'echet distance at most $r$ from $P[i,n]$.}
\label{alg:dec-k-points}
\end{algorithm}

\begin{lemma}
Let $\ell$ be a line in the plane, $k \ge 1$ an integer, and $r > 0$.
The call $\textsc
{\textup{Decision}}(1,\ell,k,r)$ returns \textsc{true} if and only if
there exist $k'\le k$ points on $\ell$, $Q=(q_1,\ldots,q_{k'})$, such that $\dF(P,Q)\le r$.
The running time is $k'$ times the cost of a call to $\prefix_\ell$ with a subcurve of~$P$, i.e., $O(k\log^2 n)$.
\end{lemma}
\begin{proof}
The time bound is obvious. Moreover, it is clear that, if the call $\textsc{Decision}(1,\ell,k,r)$ returns \textsc{true}, then
there exist $k' \le k$ such points---namely, the centers of the disks corresponding to the computed prefixes.   
We now prove the other direction by induction on~$k$.

If there exists a point $q_1$ on $\ell$ such that $\dF(P,(q_1)) \le r$, then the call $\prefix_\ell(P[1,n],r)$ at the top level of the recursion will return $n$ and the algorithm will return \textsc{true} just at the beginning of the next level (since $i=n+1$).

Assume now that for any preprocessed curve $P'$, if there exist fewer than $k$ points on~$\ell$, $Q'=(q'_1,\ldots,q'_{k-1})$, such that $\dF(P',Q') \le r$, then the call $\textsc{Decision}(1,\ell,k-1,r)$ (applied to $P'$) returns \textsc{true}. Moreover, assume there exist $k$ points on $\ell$, $Q=(q_1,\ldots,q_k)$, such that $\dF(P,Q) \le r$. We need to show that the call $\textsc{Decision}(1,\ell,k,r)$ returns \textsc{true}. Indeed, consider a walk $W$ of $P$ and $Q$ of cost $\dF(P,Q)$, and let $l'$ be the largest index such that $W$ matches $q_1$ to $p_{l'}$. By definition, the index $l$ returned by the call $\prefix_\ell(P[1,n],r)$ at the top level of the recursion is at least as large as $l'$. Now, since the cost of $W$ is $\dF(P,Q)$ and $W$ matches $q_1$ to the prefix $P[1,l']$, we know that $\dF(P[l'+1,n],(q_2,\ldots,q_k)) \le r$ and therefore there exists a sequence $Q'$ of at most $k-1$ points on $\ell$, which is $(q_2,\ldots,q_k)$ or a suffix of it, such that $\dF(P[l+1,n],Q') \le r$. Therefore, by the induction hypothesis, the call $\textsc{Decision}(l+i,\ell,k-1,r)$ at the top level of the recursion (where $i=1$) returns \textsc{true}, and thus $\textsc{Decision}(1,\ell,k,r)$ returns \textsc{true}.
\end{proof}

\begin{note*}
  A minor modification of the \textsc{Decision} function, with the same preprocessing of~$P$, can solve the problem of finding, given line~$\ell$ and distance~$r$, the shortest sequence of points~$Q$ on~$\ell$ lying within discrete Fr\'echet distance $r$ of $P$. (No such set exists if $r$ is smaller than the maximum distance from a point of $P$ to $\ell$; this can be efficiently checked with no asymptotic slowdown.)  The query can be answered in time $O(k \log^2 n)$, where $n$ is the length of $P$ and $k$ is the length of the shortest such curve $Q$.  Performing the same query on a subcurve $P[i,j]$ of $P$ can be done at the same asymptotic cost. 
\end{note*}

\subparagraph*{The optimization algorithm.} 

We now present the query algorithm. That is, given $\ell$ and $k$, find an at-most-$k$-vertex curve $Q$ on $\ell$ minimizing $\dF(P,Q)$. The function \textsc{Optimization} in Algorithm~\ref{alg:opt-k-points} actually returns the distance $\dF(P,Q)$, but it can be easily modified (within the same bounds) to return $Q$ as well.  

\begin{algorithm}[h]
\begin{algorithmic}
    \Function{Optimization}{$i,\ell,k$}
   \If{$k=0$} \Return $\infty$ \EndIf
        \State Search for the smallest index $j$, $i \le j \le n$, such that  
        \State \qquad \Call{Decision}{$i,\ell,k,\radius_\ell(P[i,j])$}
        returns \textsc{true} 
       \State $d_1 \gets \radius_\ell(P[i,j])$
        \If{$j>i$} 
            \State  $d_2 \gets$ \Call{Optimization}{$j,\ell,k-1$}    
        \Else
            \State $d_2 \gets \infty$
        \EndIf
        \State \Return $\min\{d_1,d_2\}$
\EndFunction
\end{algorithmic}
\caption{The optimization algorithm: Find the smallest distance $r$, for which there exists a sequence $Q$ of at most $k$ points on $\ell$ such that $\dF(P[i,n],Q) \le r$.}
\label{alg:opt-k-points}
\end{algorithm}

\begin{lemma}
\label{lem:opt-k-points-works}
Let $\ell$ be a line in the plane, and let $k \ge 1$. The call \Call{Optimization}{$i,\ell,k$} returns the smallest $r$, for which there exists a sequence $Q$ of at most $k$ points on $\ell$ such that $\dF(P,Q) \le r$.
The running time is $O(k^2\log^3 n + k\log^4 n)$. 
\end{lemma}

\begin{proof}
If $k=1$, then the algorithm returns $\min\limits_{q_1\in\ell}\dF(P[i,n],(q_1))=\radius_\ell(P[i,n])$. Indeed, let $j$ be the smallest index such that \Call{Decision}{$i,\ell,1,\radius_\ell(P[i,j])$} returns \textsc{true}, so the algorithm sets $d_1\gets\radius_\ell(P[i,j]) \le \radius_\ell(P[i,n])$. But since $k=1$, the latter inequality must be an equality, so $d_1 = \radius_\ell(P[i,n])$. As for $d_2$, it is set to $\infty$ either by the recursive call with $k=0$ (if $j > i$) or by the {\bf else} clause (if $j=i$). Hence, the algorithm returns $d_1 = \min\limits_{q_1\in\ell}\dF(P[i,n],(q_1))$ as claimed.

Assume now that $k>1$ and that the algorithm returns the correct value for $k-1$. Let
\[
  d\bydef\min\limits_{\substack{k'\in [1,k] \\ q_1,\ldots,q_{k'}\in \ell }}\dF(P[i,n],(q_1,\ldots,q_{k'})),
\] 
and let $Q=(q_1,\dots,q_{k'})$ be a sequence of $k'$ points on $\ell$ such that $\dF(P[i,n],Q)=d$. Finally, let $W$ be a minimum-cost walk of $P$ and $Q$ (i.e., a walk of cost $d$). 
Let $P[i,j]$ be the shortest prefix for which \Call{Decision}{$i,\ell,k,\radius_\ell(P[i,j])$} returns \textsc{true}. Then $d_1 = \radius_\ell(P[i,j])$ and by definition $d_1 \ge d$. Notice that the inductive assumption implies that the value $d_2$ returned by the recursive call (if such a call is needed) is correct. We consider the cases $j=i$ and $j > i$ separately.

If $j=i$, $\radius_\ell((p_i))$ is feasible. This radius is equal to the distance from $p_i$ to $\ell$.  Notice that, for each point $p$ of $P[i,n]$, the distance between $p$ and $\ell$ is a lower bound on $d=\radius_\ell(P[i,n])$, so we must have $d_1=d$. Moreover, since $j \not>i$, the algorithm sets $d_2$ to $\infty$ and therefore returns the correct value.

Assume therefore that $j>i$ and let $d'\bydef \radius_\ell(P[i,j-1])$. We know that the call \Call{Decision}{$i,\ell,k,\radius_\ell(P[i,j-1])$} returned \textsc{false}, so $d'< d \le d_1$. 
We distinguish between two subcases.
\begin{itemize}
    \item If $d_1=d$, we claim that the value $d_2$ returned by the recursive call must be greater or equal than $d_1$, and therefore the algorithm returns the correct value, $d_1$. Indeed, if $d_2<d_1$, then we found $k'\le k$ points $q'_1,\dots, q'_{k'}$, where $q'_1 = \cntr_\ell(P[i,j-1])$ and $q'_2,\ldots,q'_{k'}$ are the centers induced by the recursive call, such that $\dF(P[i,n],(q'_1,\ldots,q'_{k'}))\le\max\{d',d_2\}<d_1$, contradicting the assumption that $d_1=d=\min\limits_{\substack{k'\in [1,k] \\ q_1,\ldots,q_{k'}\in \ell }}\dF(P[i,n],(q_1,\ldots,q_{k'}))$.
    \item If $d_1>d$, then let $P[i,j']$ be the prefix assigned to $q_1$ by the minimum-cost walk $W$. Notice that $j' \le j-1$, since otherwise $d$ would be greater or equal than $d_1$.
    So $d$ (which as we recall is greater than $d'$) is the smallest distance between $P[j'+1,n]$ and a curve on $\ell$ of length at most $k-1$, while $d_2$ is the smallest distance of either $P[j'+1,n]$ (if $j'=j-1$) or a proper suffix of $P[j'+1,n]$ (if $j' < j-1$) and such a curve on $\ell$. Thus, $d \ge d_2$. But, since $\max\{d',d_2\}=d_2$ is feasible, we conclude that $d=d_2$ and the algorithm returns the correct value (since $d_2 = d < d_1$).
\end{itemize}

As for the running time, at each level of the recursion we perform a binary search, which entails $O(\log n)$ calls to the decision algorithm, each preceded by a call to $\disk_\ell(P[i,j])$, with the appropriate index $j$, to obtain $\radius_\ell(P[i,j])$. Since the running time of the decision algorithm is $O(k \log^2 n)$ and the running time of $\disk_\ell(P[i,j])$ is $O(\log^4 n)$, and the number of levels is $O(k)$, we conclude that the total running time is $O(k^2 \log^3 n + k\log^5 n)$. 

This bound can be slightly improved by splitting the search into two stages: First find the canonical subcurve $S_t$ that contains the point $p_j$ that we are looking for, by adding subcurves one by one.  Then, look for $p_j$ within $S_t$ by binary search using the tree of the canonical subcurve~$S_t$. The search consists of $\log n$ steps, but in each step the set of canonical subcurves to which we need to apply $\disk_\ell$ either grows by one, or we replace the last subcurve with a shorter one. Hence, $\disk_\ell$ does not need to recompute the smallest disk for all $O(m^2)$ pairs of subcurves, but only for the $O(m)$ pairs involving the new curve.
This effectively speeds up $\disk_\ell$ by a logarithmic factor, resulting in total query time of $O(k^2 \log^3 n + k\log^4 n)$.
\end{proof}

The following theorem summarizes our main result.
\begin{theorem}
  \label{thm:path}
  Let $P$ be a polygonal curve with $n$ vertices in the plane. One can preprocess $P$ in time $O(n \log^2 n)$ into a data structure of size $O(n\log n)$, so that given a line $\ell$ and a positive integer $k$, an at-most-$k$-vertex curve on $\ell$ minimizing $\dF(P,Q)$ (and the corresponding distance) can be found in $O(k^2 \log^3 n + k\log^4 n)$ time. 
\end{theorem}

\subparagraph*{A generalization.}

Observe that the logic of both the \textsc{Decision} and \textsc{Optimization} algorithms applies in any metric space, provided the following operation is available: let $S$ be a set of points, in our case always a subcurve $P[i,j]$ of the given curve $P$, and let $X$ be a subset of the space (analogous to line $\ell$ in the above description).  Define $\ball_X(S)$ to be the smallest ball enclosing $S$ with center in~$X$ and $\radius_X(S)$ to be its radius.  Let $D(n)$ be the worst-case runtime of computing this ball, where $n$ is the length of~$P$.

Given an implementation of $\radius_X$, we can compute $\prefix_X(P,r)$ (defined analogously) as follows: We find the index $i$ so that $\radius_X(P[1,i])\leq r$ and $\radius_X(P[1,i+1])>r$, by binary search on $i$.  If $\radius_X(P)\leq r$, we return $i=n$.  The running time is therefore $O(D(n)\log n)$. 
We now use $\radius_X$ and $\prefix_X$ to implement the analog of \textsc{Decision} in time $O(k D(n) \log n)$ by running $\prefix_X$ at most $k$ times. 

Finally, we implement \textsc{Optimization} as in Algorithm~\ref{alg:opt-k-points}, that is, by using \textsc{Decision} in a binary search manner at each level of the recursion, for a total of $O(k\log n)$ times.  
We also need to invoke $\radius_X$, but the calls to \textsc{Decision} dominate the cost.  We conclude with a generalization of Theorem~\ref{thm:path}, where again the computation of the actual curve $Q$ is an easy modification:
\begin{obs}
  \label{obs:path-generic}
  Let $P$ be a polygonal curve with $n$ vertices in a metric space. Provided $P$ can be preprocessed to support $D(n)$-time $\radius_X$ queries as above, an at-most-$k$-vertex curve~$Q$ on~$X$ minimizing $\dF(P,Q)$ and the corresponding distance can be found in $O(k^2 D(n) \log^2 n)$ time. Moreover, the underlying decision algorithm can be implemented in time $O(k D(n) \log n)$. 

\end{obs}

\section{Extension to geometric trees}
\label{sec:trees}

A geometric tree is a tree whose vertices are points in the plane and whose edges are line segments connecting the corresponding points. In this section, we expand the data structure to support preprocessing a geometric tree~$T$ so that, given two vertices $u,v$ of the tree, an arbitrary line $\ell$, and an integer $k>0$, one can find an at-most-$k$-vertex curve $Q$ on $\ell$ that minimizes the discrete Fr\'echet distance to the path $\Pi_{uv}$ from $u$ to $v$ in $T$.
We use ideas similar to those in~\cite{oracle-SoCG}, which in turn use heavy-path decomposition of~$T$~\cite{heavy}.

More specifically, the heavy-path decomposition of $T$ is a collection $\{\pi_i\}$ of vertex-disjoint paths in $T$ which collectively cover all the vertices of $T$.  Let $n$ be the number of vertices in~$T$.  Any path $\Pi_{uv}$ in~$T$ is a concatenation of $O(\log n)$ subpaths of (some of) the paths $\pi_i$ and,  this information can be extracted from the heavy-path decomposition data structure in time $O(\log^2 n)$ \cite{heavy} (in fact, a faster implementation is possible, but would not help us, as it is not the bottleneck here).
We then construct the data structure $\TFVD(\pi_i)$, for each curve~$\pi_i$, as above.
Since the paths~$\pi_i$ are disjoint, the total size of the data structures involved is~$O(n\log n)$.

Given a decomposition of $\Pi_{uv}$ as above, we simply run Algorithm~\ref{alg:opt-k-points} (which, in turn, invokes Algorithm~\ref{alg:dec-k-points}), on a concatenation of $l=O(\log n)$ subpaths $P_1,\dots,P_l$ of preprocessed heavy paths.  \textsc{Decision} algorithm requires an implementation of $\prefix_\ell$ for such a concatenation.  Since we can partition the subpaths further, each into $O(\log n)$ canonical subpaths, we are effectively running the $\prefix_\ell$ algorithm for a subpath, as in Section~\ref{sec:impl-prefix}, but with $m=l\cdot O(\log n)= O(\log^2 n)$ canonical subpaths, resulting in running time of $O(\log^3 n)$ for $\prefix_\ell$ and $O(k \log^3 n)$ for \textsc{Decision}.

As for \textsc{Optimization}, conceptually we run Algorithm~\ref{alg:opt-k-points} with minor modifications to speed it up.  There are at most $k$ levels of recursion.  We will describe and analyze one level and multiply the runtime by $k$.

At a fixed level of recursion, we need to find the shortest prefix of the concatenation of $O(\log n)$ subpaths $P_1,P_2,\dots,P_l$ that satisfies the \textsc{Decision} condition.  We first find $j$ so that $P_1\cdot\ldots\cdot P_j$ satisfies it, but $P_1\cdot\ldots\cdot P_{j-1}$ does not (or $j=1$, in which case we are done), adding one subpath at a time, sequentially.  This takes at most $l$ invocations of \textsc{Decision}, for a total of $O(k\log^4 n)$ time.  Now we subdivide $P_{j}$ into a logarithmic number of canonical subcurves and apply the same logic to find the canonical subcurve $P'$ that contains the (still unknown) point~$p_j$ we want.  This involves $O(\log n)$ more invocations of \textsc{Decision}, for a total of $O(k\log^4 n)$ time. 

However, we did not account for the work of the $\radius_\ell$ algorithm.  It is easily checked that during these first two steps, over all the calls to $\radius_\ell$, it is enough to compute the minimum enclosing radius with center on $\ell$ for at most ${m \choose 2}=O(\log^4 n)$ pairs of canonical subpaths constituting the subpaths $P_1,\dots,P_l$, at a cost of $O(\log^2 n)$ per pair; so the total cost of all the $\radius_\ell$ runs at the current level of recursion up to now is at most $O(\log^6 n)$.

Finally, we use the hierarchical decomposition of $P'$ to find desired point $p_j$, essentially as in the proof of Theorem~\ref{thm:path}.  We observe that descending the hierarchy in binary search, we always either add the last canonical subpath or replace the last canonical subpath, so the number of pairs of subpaths that need to be recomputed by the $\radius_\ell$ algorithm on each descent step is only $O(m)=O(\log^2n)$  (rather than $O(m^2)$) and therefore each $\radius_\ell$ call takes $O(m\log^2 n)=O(\log^4 n)$ time, for a total of $O(\log^5n)$ over the entire binary search.  

To summarize, each level of recursion in this implementation of \textsc{Optimization} takes $O(k\log^4 n + \log^6 n)$ time.
We conclude that that the full running time of \textsc{Optimization} is $O(k^2\log^4 n + k\log^6 n)$, which finishes the proof of the following theorem:

\begin{theorem}
  \label{thm:tree}
  Let $T$ be a geometric tree on $n$ vertices. 
  One can preprocess $T$ in time $O(n \log^2 n)$ into a data structure of size $O(n\log n)$, so that given vertices $u$ and $v$ in $T$, a line $\ell$, and a positive integer $k$, the at-most-$k$-vertex curve $Q$ on $\ell$ minimizing the discrete Fr\'echet distance between $Q$ and the path $\Pi_{uv}$ between $u$ and $v$ in $T$, can be found in $O(k^2\log^4 n + k\log^6 n)$ time. 
\end{theorem}

\section{Extension to polygonal regions in the plane}
\label{sec:region}

Let $R$ be a (topologically closed) polygonal region in the plane bounded by $m$ edges and $P$~be an $n$-vertex curve. Given an integer $k>0$, we show how to find an at-most-$k$-vertex curve $Q\subseteq R$ that minimizes the discrete Fr\'echet distance $d_{dF}(P,Q)$. That is, we find a curve $Q =(q_1,\ldots,q_{k'})$, with $k' \le k$ and $q_1,\ldots,q_{k'} \in R$, that realizes the expression
\begin{equation*}
\min\limits_{\substack{k'\in [1,k] \\ q_1,\ldots,q_{k'}\in R}}
    \{\dF(P,(q_1,\ldots,q_{k'}))\}.
\end{equation*}
Note that $Q$ need not be unique.

Observe that we only need to implement one new primitive, which we refer to as $\disk_R(S)$, that computes the smallest-radius disk centered at a point of $R$ for a subcurve $S$ of $P$ (and its corresponding $\radius$ and $\cntr$). Using $\disk_R(S)$ one can perform binary search over~$P$ to implement $\prefix_R(S,r)$ at the cost of a logarithmic factor in the running time. The \textsc{decision} and \textsc{optimization} functions (Algorithms~\ref{alg:dec-k-points} and~\ref{alg:opt-k-points}) will use the new primitive and otherwise remain the same.

We now sketch how to implement $\disk_R(S)$ for a subcurve $S=P[i,j]$. Let $c$ be the center of the mec containing the subcurve $P[i,j]$ and centered at a point of $R$.  

First we find the center of the unconstrained mec of $P[i,j]$, using the data structure in \cite{eppstein1992dynamic} for the mec of a planar point set. We partition $P[i,j]$ into $t=O(\log n)$ canonical subcurves (as in Section~\ref{sec:data}) and preprocess each subcurve as in \cite{eppstein1992dynamic}.  Using the algorithm from Section~3 of \cite{eppstein1992dynamic}, we can compute the mec of the union of the subcurves, that is, of $P[i,j]$ in deterministic time $O(t^3 \log^3n) =O(\log^6 n)$ (refer to the discussion after Lemma~2 in that paper).  Alternatively, Theorem~2 in \cite{eppstein1992dynamic} provides an algorithm for finding the mec in expected time $O(t \log t \log n + \sqrt{t} \log t \log^3 n)=O(\log^{3.5}n \log\log n)$.\footnote{%
  The statements are phrased in terms of solving an LP over an intersection of polyhedra in $\reals^3$, but a later discussion points out that the same logic, with minor modifications, applies to finding the mec of a discrete point set in the plane as well \cite{eppstein1992dynamic}.}
If the center of the resulting mec lies in $R$, we are done: we found $\disk_R(S)$.

Otherwise, since $d_{\max}(P[i,j],x)$ is a convex function of $x$, the center~$c$ lies on the boundary of~$R$.  For each bounding segment $s$ of $R$, we compute $\disk_\ell(P[i,j])$ for the supporting line~$\ell$ of~$s$ in time $O(t^2\log^2n)=O(\log^4 n)$; refer to Lemma~\ref{lem:union_of_m}.  Once again, if the center $c(s)\coloneqq\cntr_\ell(P[i,j])$ lies in~$s$, this gives a candidate mec, otherwise the endpoint of $s$ closest to $c(s)$ provides such a candidate.  After repeating the process for each boundary edge of~$R$, we return (the center of) the smallest disk found.  The total cost is thus $O(m\log^4 n)$.

To summarize, $\disk_{R}(P[i,j])$ can be computed in deterministic time $O(\log^6n+m\log^4 n)$ or in expected time $O(\log^{3.5}n \log\log n+m\log^4 n)=O(m\log^4 n)$, assuming $m\neq0$.

We now implement $\prefix_R$ by binary search over the prefix length using $\disk_R(S)$ in expected time $O(m\log^5 n)$.
Expected running time of \textsc{decision} will be $O(km\log^5 n)$ and that of \textsc{optimization} will be $O(k^2m\log^6 n)$ (see the time analysis of Lemma~\ref{lem:opt-k-points-works} for details). 

We did not optimize the logarithmic factors.

\begin{theorem}
\label{thm:region}
  Let $P$ be a polygonal curve with $n$ vertices in the plane. One can preprocess $P$ in time $O(n \log^2 n)$ into a data structure of size $O(n\log n)$, so that given a polygonal domain~$R$ bounded by $m$ segments, an at-most-$k$-vertex curve $Q\subseteq R$ that minimizes the discrete Fr\'echet distance $d_{dF}(P,Q)$ and the corresponding distance can be found in expected time $O(k^2m\log^6 n)$.  Alternatively, it can be found in deterministic time $O(k^2 m \log^6 n + k^2 \log^8 n)$. 
\end{theorem}

\section{Discrete Fr\'echet distance simplification on a \texorpdfstring{$\dm$}{g}-flat}
\label{sec:approx}

In this section we address the general version of the dimension-reduction problem mentioned in the introduction. That is, $P$ is a polygonal curve in $\reals^d$, for some constant $d > 2$.  We need to preprocess $P$ into a compact data structure to efficiently support queries that specify a $\dm$-flat $h$ ($1 \le \dm \le d-1$) and an integer $k \ge 1$, and seek a curve $Q$ on~$h$ of length at most $k$ that minimizes the discrete Fr{\'e}chet distance to $P$. 

It is unlikely that there exists a solution to this problem with bounds similar to those that we obtained in Section~\ref{sec:path}, since it seems that, on the one hand, a primitive analogous to $\disk_\ell(Q)$, where `disk' is replaced by `ball' and $\ell$ is replaced by $h$, is essential, but on the other hand, it is probably impossible to implement such a primitive efficiently in higher dimensions, i.e., in polylogarithmic time after near-linear time preprocessing of $Q$. 

In passing, we mention that (while it is difficult or maybe even impossible to implement $\ball_h(P[i,j])$ efficiently in the $L_2$ metric) it \emph{is} realizable in the $L_\infty$ metric as, based on Observation~\ref{obs:path-generic}, we have:

\begin{note*}
  For the $L_\infty$ interpoint distance in $\reals^d$, $D(n)$ is $O_d(\log n)$, with a constant depending on the dimension~$d$, since one can compute the bounding box of $P[i,j]$ in logarithmic time and the bounding box together with $h$ determine the smallest enclosing $L_\infty$-ball with a center on~$h$; the running time of this final part of $\ball_h(P[i,j])$ depends only on~$d$.
\end{note*}

Therefore, since we are interested in a solution to the general dimension-reduction problem with near-linear time preprocessing and polylogarithmic time query, we resort to approximate queries. Specifically, we present such a solution (for a prespecified parameter $0 < \eps_0 < 1$) that given $h$ returns a curve $Q$ on~$h$ of length at most $k$ that minimizes the discrete Fr{\'e}chet distance to $P$ up to a factor of $1+\eps_0$ (i.e., 
$\dF(P,Q) \le (1+\eps_0)\dF(P,Q^*)$, where $Q^*$ is such a curve that minimizes the discrete Fr\'echet distance to $P$).

Set $\eps \coloneqq \frac{\eps_0}{1+\eps_0}$; then, $0 \le \eps \le 1/2$ and
$\frac{1}{1-\eps}=1+\eps_0$. 
We first define the primitive $\ball_h^\eps(Q)$, which returns a $(1-\eps)$-approximation of the smallest ball centered at a point of $h$ and containing a set $Q$ of points. That is, let $\radius_h^\eps(Q)$ denote the radius of the returned ball, then $(1-\eps)\radius_h(Q) \le \radius_h^\eps(Q) \le \radius_h(Q)$, where $\radius_h(Q)$ is the radius of the smallest ball with center on~$h$.

We use \emph{coresets} to implement $\ball_h^\eps(Q)$, assuming $Q$ is the set of points corresponding to a subsequence of $P$. More precisely, we preprocess $P$ so that given a subsequence $P[i,j]$, for $1 \le i \le j \le n$, one can compute $\ball_h^\eps(P[i,j])$ in $O(c_\eps \log n)$ time, where $c_\eps = c_\eps(\eps)$ is a constant, see below. Before presenting the details,
we provide a brief overview of the coreset-related definitions and results that we use.

\subparagraph*{Coresets.}
Let $S$ be a compact set in $\reals^d$.  Let $\theta\in \sphere^{d-1}$ be a direction.  The \emph{directional width} $\width(S,\theta)$ of $S$ in direction $\theta$ is the minimum distance between a pair of hyperplanes with normal $\theta$ containing $S$ between them.  Fix a number $\eps$, $0<\eps<1$.  A subset $C$ of $S$ is an \emph{$\eps$-coreset} for $S$ \emph{with respect to directional width}, if it has the property that $\width(C,\theta)\ge (1-\eps)\width(S,\theta)$, for all directions $\theta$; since $C\subset S$, $\width(C,\theta)\le \width(S,\theta)$.  We say that a subset $C$ of $S$ is an \emph{$\eps$-coreset} for $S$ for \emph{farthest neighbors} if, for any point $q$, $d_{\max}(C,q) \ge (1-\eps)d_{\max}(S,q)$. It follows that if $B$ is any ball containing $C$, then the ball $\frac{B}{1-\eps}$, obtained by scaling $B$ around its center by a factor of $1/(1-\eps)$, contains~$S$.    

The following facts are known (refer for example to~\cite[Section~23.1]{sariel-book}):
\begin{itemize}
\item A compact set in $\reals^d$ has an $\eps$-coreset for directional width of size $c_\eps\coloneqq O(1/\eps^{(d-1)/2})$.
\item Such an $\eps$-coreset for a set of $n$ points can be constructed in time $O(n+c_\eps^3)$. %
\item An $(\eps/4)$-coreset for directional width is an $\eps$-coreset for farthest neighbors.
\item If $C_1,C_2,\dots,C_t$ are $\eps$-coresets for $S_1,S_2,\dots,S_t$ for width or farthest neighbors, then $\bigcup C_i$ is an $\eps$-coreset for $\bigcup S_i$ for width or farthest neighbors.
\end{itemize}

As we are interested in coresets for farthest neighbors, with a slight abuse of terminology, we will use unqualified ``coreset'' to denote them hereafter.

\subparagraph*{Smallest ball with a center on a flat.}
We recall the classical randomized LP-type algorithm of Welzl~\cite{emo-mec} for finding the smallest enclosing ball of an $n$-point set $S$ in $\reals^d$ in $O_d(n)$ expected time (with the implied constant depending on $d$).  An inspection of the algorithm reveals that it can be used  essentially unmodified to compute the smallest enclosing ball of $S$ whose center is constrained to lie in a given $g$-flat; such a ball will be defined by at most $g+1$ points.

\subparagraph*{Implementation of $\ball_h^\eps(P[i,j])$.}
We now explain how to implement $\ball_h^\eps(P[i,j])$, which is the basic building block of the approximation algorithm described above.  As in Section~\ref{sec:data} we once again construct a hierarchical decomposition of $P$, storing $P$ of length $n$, then $P[1,n/2]$ and $P[n/2+1,n]$ as its children and so forth, thereby creating a hierarchy of canonical subcurves.  With each canonical subcurve $P[i,j]$ we associate its coreset $C_{ij}$ for farthest neighbors, of size at most $c_\eps$.  An arbitrary subcurve $P'$ of $P$ is the union of $O(\log n)$ canonical subcurves. If we collect the points of the corresponding coresets, we obtain a coreset $C'$ for $P'$ of size $c_\eps \log n$.  Running the exact minimum-enclosing-ball with center on a $g$-flat algorithm on the set $C'$, by the above discussion, gives an $\eps$-approximation of corresponding ball for $P'$. Indeed, we obtain the smallest possible ball $B$ with the center constrained to the flat and enclosing~$C'$.  By the coreset property $\frac{B}{1-\eps}$ contains~$P'$.  There cannot be a substantially smaller ball with center on the flat and containing~$P'$, as such a ball would also contain~$C'$ and would contradict minimality of $B$.

The expected running time is dominated by that of the ball-finding routine, which is $O_d(c_\eps \log n)=O_d(\log n/\eps^{(d-1)/2})$, concluding our description of an implementation of $\ball_h^\eps$.

\subparagraph*{The decision algorithm.}
We now wish to devise a decision algorithm, similar to the one in Section~\ref{sec:path}. Given $r > 0$, our decision algorithm will return \textsc{true} if there exists a sequence $Q$ of at most $k$ points on~$h$ such that $\dF(P,Q) \le \frac{r}{1-\eps}$. It will return \textsc{false} otherwise, in which case we may only conclude that for any sequence $Q$ of at most $k$ points on $h$, $\dF(P,Q) > r$.

For this purpose, we define the primitive $\prefix_h^\epsilon(S,r)$ for a sequence $S$ of points and a radius $r$. 
It returns a prefix $S[1,k]$ of $S$, such that $\ball_h^\eps(S[1,k])$ returns a radius that is at most $r$ where as $\ball_h^\eps(S[1,k+1])$ returns a radius that is greater than $r$. In practice, $\prefix_h^\eps(S,r)$ returns the length of the prefix, so 0 if it is empty. Thus, if $\prefix_h^\eps(S,r)=k$, $0 < k < n$, then $\radius_h(S[1,k]) \le \frac{r}{1-\eps}$ and $\radius_h(S[1,k+1]) > r$. (If $\prefix_h^\eps(S,r) = 0$, then $\radius_h(S[1,1]) > r$, and if $\prefix_h^\eps(S,r) = n$, then $\radius_h(S[1,n]) \le \frac{r}{1-\eps}$.) 

\begin{observation}
\label{obs:prefix}
    If $\prefix_h(S,r)$ returns the largest $k'$ for which $\radius_h(S[1,k']) \le  r$, then   
    $\prefix_h(S,r) \le \prefix_h^\eps(S,r) \le \prefix_h(S,\frac{r}{1-\eps})$.
\end{observation}

To implement $\prefix_h^\epsilon(S,r)$ we perform a binary search. That is, we set $i = \lfloor \frac{|S|}{2} \rfloor$ and call $\ball_h^\eps(S[1,i])$. Now, if $\radius_h^\eps(S[1,i]) \le r$, then we increase~$i$, and otherwise we decrease $i$, etc. Thus, the cost of a call to $\prefix_h^\epsilon$ with a subcurve of $P$ is $O(c_\eps\log^2 n)$ time.

We now present the decision algorithm (see Algorithm~\ref{alg:approx_dec-k-points}).

\begin{algorithm}[h]
\begin{algorithmic}
    \Function{ApproxDecision}{$i,h,k,r$}
        \If{$i=n+1$}
            \State return \textsc{true}
        \EndIf 
        \If{$k=0$}
            \State return \textsc{false}
        \EndIf
        \State $l\gets \prefix_h^\epsilon(P[i,n],r)$
        \If{$l=0$}  \Comment{$p_i$ is too far from $h$}
            \State return \textsc{false}
        \EndIf
        \State return \Call{ApproxDecision}{$i+l,h,k-1,r$}
    \EndFunction
    \end{algorithmic}
\caption{The decision algorithm: Given a distance $r > 0$, determine if there exists a sequence of at most $k$ points on $h$ at discrete Fr\'echet distance at most $\frac{r}{1-\eps}$ from $P[i,n]$.}
\label{alg:approx_dec-k-points}
\end{algorithm}

The following lemma follows from Observation~\ref{obs:prefix}.
\begin{lemma}
Let $h$ be a $\dm$-dimensional flat in $\reals^d$, $k \ge 1$ an integer, and $r > 0$. If the call $\textsc
{\textup{AproxDecision}}(1,h,k,r)$ returns \textsc{true}, then there exists a sequence $Q$ of $k'\le k$ points on $h$, such that $\dF(P,Q)\le \dfrac{r}{1-\eps}$. If it returns \textsc{false}, then for any sequence $Q$ of $k'\le k$ points on $h$, $\dF(P,Q)> r$.
The running time is $k'$ times the cost of a call to $\prefix_h^\eps$ with a subcurve of~$P$.
\end{lemma}

Finally, we run the optimization algorithm in Section~\ref{sec:path}, using the decision algorithm above (and replacing $\radius_\ell$ by $\radius_h^\eps$). 
The following theorem summarizes the section's main result.
\begin{theorem}
  \label{thm:path_higher}
  Let $P$ be a polygonal curve with $n$ vertices in $\reals^d$. Let $0 < \eps_0 < 1$, and set $\eps = \frac{\eps_0}{1+\eps_0}$. One can preprocess $P$ in time $O(n\log n + \frac{n}{\eps^{3(d-1)/2}})$ into a data structure of size $O(n\log n + \frac{n}{\eps^{(d-1)/2}})$, so that given a $g$-dimensional flat $h$ and a positive integer $k$, an at-most-$k$-vertex curve $Q$ on $h$ such that $\dF(P,Q) \le \frac{\dF(P,Q^*)}{1-\eps}=(1+\eps_0)\dF(P,Q^*)$ (and the corresponding distance) can be found in $O(\frac{k^2\log^2 n}{\eps^{(d-1)/2}})$ time, where $Q^*$ is such a curve that minimizes the discrete Fr\'echet distance to $P$.
\end{theorem}

\section{Discussion and open problems}

We mention several open problems:
\begin{itemize}
\item We did not strive to optimize the query time. How far can it be reduced, without substantially increasing the space and preprocessing time, in Theorems~\ref{thm:path}, \ref{thm:tree}, \ref{thm:region}, and~\ref{thm:path_higher}?
\item Can the factor $k^2$ in the query time in these theorems be made linear in $k$?  Or at least subquadratic?
\item Several of the methods we have described can be combined in straightforward ways.  For example, we can solve the geometric tree version of the problem exactly in $\reals^d$ for $L_\infty$ interpoint distances, or approximately for $L_2$, by combining the tools from Section~\ref{sec:trees} and Section~\ref{sec:approx}. 
\item Further, we can obtain an approximate solution in $\reals^d$ with the simplified curve~$Q$ restricted to lie in a polyhedral region~$R$ (of any dimension up to and including~$d$) by observing that our approximating function is in fact the actual smallest-enclosing ball radius function for a coreset and thus convex, so an approximation can be computed by using this function instead.  This combines the tools from Section~\ref{sec:region} and Section~\ref{sec:approx}.  If $R$ is bounded by $m$ faces of all dimensions, all simplicial, the cost of the query would be proportional to~$m$.  In the special case when~$R$ is a convex polyhedron with~$f$ facets, the processing can be conceivably sped up as a function of~$f$, as when solving a linear program, though we did not investigate this direction further at this time.
\end{itemize}

\bibliographystyle{abbrvurl}
\bibliography{refs}

\begin{thebibliography}{10}

\bibitem{agarwal2017range}
P.~K. Agarwal.
\newblock Range searching.
\newblock In {\em Handbook of Discrete and Computational Geometry}, pages
  1057--1092. Chapman and Hall/CRC, 2017.

\bibitem{AgarwalHMW05}
P.~K. Agarwal, S.~Har{-}Peled, N.~H. Mustafa, and Y.~Wang.
\newblock Near-linear time approximation algorithms for curve simplification.
\newblock {\em Algorithmica}, 42(3-4):203--219, 2005.
\newblock URL: \url{https://doi.org/10.1007/s00453-005-1165-y}, \href
  {https://doi.org/10.1007/S00453-005-1165-Y}
  {\path{doi:10.1007/S00453-005-1165-Y}}.

\bibitem{oracle-SoCG}
B.~Aronov, T.~Farhana, M.~J. Katz, and I.~Ramesh.
\newblock Discrete {F}r\'{e}chet distance oracles.
\newblock In W.~Mulzer and J.~M. Phillips, editors, {\em 40th International
  Symposium on Computational Geometry (SoCG 2024)}, pages 10:1--10:14, 2024.
\newblock URL:
  \url{https://drops.dagstuhl.de/entities/document/10.4230/LIPIcs.SoCG.2024.10},
  \href {https://doi.org/10.4230/LIPIcs.SoCG.2024.10}
  {\path{doi:10.4230/LIPIcs.SoCG.2024.10}}.

\bibitem{BeregJWYZ08}
S.~Bereg, M.~Jiang, W.~Wang, B.~Yang, and B.~Zhu.
\newblock Simplifying 3{D} polygonal chains under the discrete {F}r{\'{e}}chet
  distance.
\newblock In {\em {LATIN} 2008: Theoretical Informatics, 8th Latin American
  Symposium, B{\'{u}}zios, Brazil, April 7-11, 2008, Proceedings}, volume 4957
  of {\em Lecture Notes in Computer Science}, pages 630--641. Springer, 2008.
\newblock \href {https://doi.org/10.1007/978-3-540-78773-0_54}
  {\path{doi:10.1007/978-3-540-78773-0_54}}.

\bibitem{query-BLR08}
P.~Bose, S.~Langerman, and S.~Roy.
\newblock Smallest enclosing circle centered on a query line segment.
\newblock In {\em Proceedings of the 20th Annual Canadian Conference on
  Computational Geometry, Montr{\'{e}}al, Canada, August 13-15, 2008}, 2008.

\bibitem{aggr13}
P.~Brass, C.~Knauer, C.~Shin, M.~H.~M. Smid, and I.~Vigan.
\newblock Range-aggregate queries for geometric extent problems.
\newblock In A.~Wirth, editor, {\em Nineteenth Computing: The Australasian
  Theory Symposium, {CATS} 2013, Adelaide, Australia, February 2013}, volume
  141 of {\em {CRPIT}}, pages 3--10. Australian Computer Society, 2013.
\newblock URL:
  \url{http://crpit.scem.westernsydney.edu.au/abstracts/CRPITV141Brass.html}.

\bibitem{BringmannC20}
K.~Bringmann and B.~R. Chaudhury.
\newblock Polyline simplification has cubic complexity.
\newblock {\em J. Comput. Geom.}, 11(2):94--130, 2020.
\newblock URL: \url{https://doi.org/10.20382/jocg.v11i2a5}, \href
  {https://doi.org/10.20382/JOCG.V11I2A5} {\path{doi:10.20382/JOCG.V11I2A5}}.

\bibitem{DriemelH13}
A.~Driemel and S.~Har{-}Peled.
\newblock Jaywalking your dog: Computing the {F}r{\'{e}}chet distance with
  shortcuts.
\newblock {\em {SIAM} J. Comput.}, 42(5):1830--1866, 2013.
\newblock \href {https://doi.org/10.1137/120865112}
  {\path{doi:10.1137/120865112}}.

\bibitem{eppstein1992dynamic}
D.~Eppstein.
\newblock Dynamic three-dimensional linear programming.
\newblock {\em ORSA Journal on Computing}, 4(4):360--368, 1992.

\bibitem{FanFKWZ15}
C.~Fan, O.~Filtser, M.~J. Katz, T.~Wylie, and B.~Zhu.
\newblock On the chain pair simplification problem.
\newblock In F.~Dehne, J.~Sack, and U.~Stege, editors, {\em Algorithms and Data
  Structures - 14th International Symposium, {WADS} 2015, Victoria, BC, Canada,
  August 5-7, 2015. Proceedings}, volume 9214 of {\em Lecture Notes in Computer
  Science}, pages 351--362. Springer, 2015.
\newblock \href {https://doi.org/10.1007/978-3-319-21840-3_29}
  {\path{doi:10.1007/978-3-319-21840-3_29}}.

\bibitem{Filtser18}
O.~Filtser.
\newblock Universal approximate simplification under the discrete
  {F}r{\'{e}}chet distance.
\newblock {\em Inf. Process. Lett.}, 132:22--27, 2018.
\newblock \href {https://doi.org/10.1016/j.ipl.2017.10.002}
  {\path{doi:10.1016/j.ipl.2017.10.002}}.

\bibitem{GuibasHMS93}
L.~J. Guibas, J.~Hershberger, J.~S.~B. Mitchell, and J.~Snoeyink.
\newblock Approximating polygons and subdivisions with minimum link paths.
\newblock {\em Int. J. Comput. Geom. Appl.}, 3(4):383--415, 1993.
\newblock \href {https://doi.org/10.1142/S0218195993000257}
  {\path{doi:10.1142/S0218195993000257}}.

\bibitem{aggr14}
P.~Gupta, R.~Janardan, Y.~Kumar, and M.~H.~M. Smid.
\newblock Data structures for range-aggregate extent queries.
\newblock {\em Comput. Geom.}, 47(2):329--347, 2014.
\newblock URL: \url{https://doi.org/10.1016/j.comgeo.2009.08.001}, \href
  {https://doi.org/10.1016/J.COMGEO.2009.08.001}
  {\path{doi:10.1016/J.COMGEO.2009.08.001}}.

\bibitem{sariel-book}
S.~Har-Peled.
\newblock {\em Geometric approximation algorithms}.
\newblock Number 173 in Mathematical Surveys and Monographs. American
  Mathematical Soc., 2011.

\bibitem{circle-query-ipl-KD08}
A.~Karmakar, S.~Roy, and S.~Das.
\newblock Fast computation of smallest enclosing circle with center on a query
  line segment.
\newblock {\em Inf. Process. Lett.}, 108(6):343--346, 2008.
\newblock URL: \url{https://doi.org/10.1016/j.ipl.2008.07.002}, \href
  {https://doi.org/10.1016/J.IPL.2008.07.002}
  {\path{doi:10.1016/J.IPL.2008.07.002}}.

\bibitem{constrained-cirle-RoyKDN09}
S.~Roy, A.~Karmakar, S.~Das, and S.~C. Nandy.
\newblock Constrained minimum enclosing circle with center on a query line
  segment.
\newblock {\em Comput. Geom.}, 42(6-7):632--638, 2009.
\newblock URL: \url{https://doi.org/10.1016/j.comgeo.2009.01.002}, \href
  {https://doi.org/10.1016/J.COMGEO.2009.01.002}
  {\path{doi:10.1016/J.COMGEO.2009.01.002}}.

\bibitem{heavy}
D.~D. Sleator and R.~E. Tarjan.
\newblock A data structure for dynamic trees.
\newblock {\em J. Comput. Syst. Sci.}, 26(3):362--391, 1983.
\newblock \href {https://doi.org/10.1016/0022-0000(83)90006-5}
  {\path{doi:10.1016/0022-0000(83)90006-5}}.

\bibitem{KerkhofKLMW19}
M.~van~de Kerkhof, I.~Kostitsyna, M.~L{\"{o}}ffler, M.~Mirzanezhad, and
  C.~Wenk.
\newblock Global curve simplification.
\newblock In M.~A. Bender, O.~Svensson, and G.~Herman, editors, {\em 27th
  Annual European Symposium on Algorithms, {ESA} 2019, September 9-11, 2019,
  Munich/Garching, Germany}, volume 144 of {\em LIPIcs}, pages 67:1--67:14.
  Schloss Dagstuhl - Leibniz-Zentrum f{\"{u}}r Informatik, 2019.
\newblock URL: \url{https://doi.org/10.4230/LIPIcs.ESA.2019.67}, \href
  {https://doi.org/10.4230/LIPICS.ESA.2019.67}
  {\path{doi:10.4230/LIPICS.ESA.2019.67}}.

\bibitem{KreveldLW20}
M.~J. van Kreveld, M.~L{\"{o}}ffler, and L.~Wiratma.
\newblock On optimal polyline simplification using the {H}ausdorff and
  {F}r{\'{e}}chet distance.
\newblock {\em J. Comput. Geom.}, 11(1):1--25, 2020.
\newblock URL: \url{https://doi.org/10.20382/jocg.v11i1a1}, \href
  {https://doi.org/10.20382/JOCG.V11I1A1} {\path{doi:10.20382/JOCG.V11I1A1}}.

\bibitem{emo-mec}
E.~Welzl.
\newblock Smallest enclosing disks (balls and ellipsoids).
\newblock In H.~A. Maurer, editor, {\em New Results and New Trends in Computer
  Science, Graz, Austria, June 20-21, 1991, Proceedings [on occasion of H.
  Maurer's 50th birthday]}, volume 555 of {\em Lecture Notes in Computer
  Science}, pages 359--370. Springer, 1991.
\newblock URL: \url{https://doi.org/10.1007/BFb0038202}, \href
  {https://doi.org/10.1007/BFB0038202} {\path{doi:10.1007/BFB0038202}}.

\end{thebibliography}

\end{document}